\documentclass[a4paper,USenglish]{lipics-v2019}

\usepackage{hyperref}
\title{Distributed Quantum Proofs for Replicated Data}

\author{Pierre Fraigniaud}{IRIF, CNRS and Universit\'e de Paris, France}{}{}{Additional support from the ANR project \emph{DESCARTES}.}
\author{Fran\c{c}ois Le Gall}{Graduate School of Mathematics, Nagoya University, Japan}{}{}{JSPS KAKENHI grants Nos.~JP16H01705, JP19H04066, JP20H04139 and JP20H00579 and MEXT Quantum Leap Flagship Program Grant Number JPMXS0120319794.}
\author{Harumichi Nishimura}{Graduate School of Informatics, Nagoya University, Japan}{}{}{JSPS KAKENHI grants Nos.~JP16H01705, JP19H04066 
	and MEXT Quantum Leap Flagship Program Grant Number JPMXS0120319794.}
\author{Ami Paz}{Faculty of Computer Science, Universit\"at Wien, Austria}
{}{}{Supported by the Austrian Science Fund (FWF): P 33775-N, Fast Algorithms for a Reactive Network Layer.}

\authorrunning{P.\ Fraigniaud, F.\ Le Gall, H.\ Nishimura and A.\ Paz} 

\Copyright{Pierre Fraigniaud, Fran\c{c}ois Le Gall, Harumichi Nishimura and Ami Paz} 



\acknowledgements{We thank an anonymous reviewer of ITCS 2021 for pointing us to~\cite{Rosgen08}.
}

\nolinenumbers 

\hideLIPIcs  

\usepackage{amsthm,amssymb,amsfonts}
\usepackage{xcolor}
\newtheorem*{result}{Main Result}

\newcommand{\ket}[1]{\lvert #1 \rangle}

\newcommand{\tr}{\operatorname{tr}}

\newcommand{\EQ}{\textsf{EQ}}
\newcommand{\HAM}{\textsf{HAM}}
\newcommand{\dAM}{\textsf{dAM}}
\newcommand{\dMA}{\textsf{dMA}}
\newcommand{\dQMA}{\textsf{dQMA}}
\newcommand{\K}{z}

\newenvironment{psmallmatrix}{\left(\begin{smallmatrix}} {\end{smallmatrix}\right)}

\begin{document}

\date{}
\maketitle

\begin{abstract}
	This paper tackles the issue of \emph{checking} that all copies of a large data set replicated at several nodes of a network are identical. The fact that the replicas may be located at distant nodes prevents the system from verifying their equality locally, i.e., by having each node consult only nodes in its vicinity. On the other hand, it remains possible to assign \emph{certificates} to the nodes, so that verifying the consistency of the replicas can be achieved locally. 
	However, we show that, as the replicated data is large, classical certification mechanisms, including distributed Merlin-Arthur protocols, cannot guarantee good completeness and soundness simultaneously, unless they use very large certificates.
	The main result of this paper is a distributed \emph{quantum} Merlin-Arthur protocol enabling the nodes to collectively check the consistency of the replicas, based on small certificates, and in a single round of message exchange between neighbors, with short messages. 
	In particular, the certificate-size is logarithmic in the size of the data set, which gives an exponential advantage over classical certification mechanisms. 
	We propose yet another usage of a fundamental quantum primitive, 
	called the SWAP test, in order to show our main result.
\end{abstract}

%


\section{Introduction} 

In the context of distributed systems, the presence of faults potentially corrupting the individual states of the nodes creates a need to regularly check whether the system is in a global state that is legal with respect to its specification. A basic example is a system storing data, and using replicas  in order to support crash failures. In this case, the application managing the data is in charge of regularly checking that the several replicas of the same data, stored at different nodes scattered in the network, are all identical. Another example is an application maintaining a tree spanning the nodes of a network, e.g., for multicast communication. In this case, every node stores a pointer to its parent in the tree, and the application must regularly check that the collection of pointers forms a spanning tree. This paper addresses the issue of checking the correctness of a distributed system configuration at low cost. 

Several mechanisms have been designed for certifying the correctness of the global state of a system in a distributed manner. One popular mechanism is called \emph{locally checkable proofs}~\cite{GoosS16}, and it extends the seminal concept of \emph{proof-labeling schemes}~\cite{KormanKP10}. In these frameworks, the distributed application does not only construct or maintain some distributed data structure (e.g., a spanning tree), but also constructs a distributed \emph{proof} that the data structure is correct. This proof has the form of a \emph{certificate} assigned to each node (the certificates assigned to different nodes do not need to be the same). For collectively checking the legality of the current global system state, the nodes exchange their certificates with their neighbors in the network. Then, based on its own individual state, its certificate, and the certificates of its neighbors, every node accepts or rejects, according to the following specification. If the global state is legal, and if the certificates are assigned properly by the application, then all nodes accept. Conversely, if the global state is illegal, then at least one node rejects, \emph{no matter  which certificates are assigned to the nodes}. Such a rejecting node can raise an alarm, or launch a recovery procedure. The main aim of locally checkable proofs is to be \emph{compact}, that is, to use certificates as small as possible, for two reasons: first, to limit the space complexity at each node, and, second, to limit the message complexity of the verification procedure involving communications between neighbors.  

For instance, in the case of the Spanning Tree predicate, the application does not only construct a spanning tree $T$ of the network, but also a distributed proof that $T$ is indeed a spanning tree, i.e., that the collection~$T$ of pointers forms a cycle-free connected spanning subgraph. It has been known for long~\cite{AfekKY97,AwerbuchPV91,ItkisL94} that, by assigning to every node a certificate of logarithmic size, the nodes can collectively check whether $T$ is indeed a spanning tree, in a single round of communication between neighboring nodes. The certificate assigned to a node is the identity of the root of the tree, and its distance to this root (both are of logarithmic size as long as the IDs are in a range polynomial in the number of nodes). Every node just checks that it is provided with the same root-ID as all its neighbors in the network, and that the distance given to its parent in its certificate is one less than its own given distance --- a node with distance~0 checks that its ID is indeed the root-ID provided in its certificate. Obviously, if the collection $T$ of pointers forms a spanning tree, and if the certificates are assigned properly by the application, then all nodes pass these tests, and accept.  On the other hand, it is easy to check that if $T$ is not a spanning tree (it is not connected, or it contains a cycle), then at least one node detects a discrepancy and rejects, no matter which certificates are assigned to the nodes. 

Unfortunately, not all boolean predicates on labeled graphs can be distributedly certified using certificates as small as for spanning tree. This is typically the case of the aforementioned scenario of a distributed data storage using replicas, for which one must certify equality. Let us for instance consider the case of two nodes Alice and Bob at the two extremities of a path, that is, the two players are separated by intermediate nodes. Alice and Bob respectively store two $n$-bit strings $x$ and  $y$, and the objective is to certify that $x=y$. That is, one wants to certify equality (\EQ) between \emph{distant} players.  A direct reduction from the non-deterministic communication complexity of \EQ\/ shows that certifying \EQ\/ cannot be achieved with certificates smaller than $\Omega(n)$ bits. 

Randomization may help circumventing the difficulty of certifying some boolean predicates on labeled graphs using small certificates. Hence, a weaker form of protocols has been considered, namely \emph{distributed Merlin-Arthur} protocols (\dMA), a.k.a.~\emph{randomized proof-labeling schemes}~\cite{FraigniaudPP19}. In this latter context, Merlin provides the nodes with a proof, just like in locally checkable proofs, and Arthur performs a \emph{randomized} local verification at each node. Unfortunately, some predicates remain hard in this framework too. In particular, as we show in the paper, there are no classical \dMA\/ protocols for (distant) \EQ\/ using compact certificates.
Recently, several extensions of $\dMA$ protocols were proposed, e.g., by allowing more interaction between the prover and the verifier~\cite{CrescenziFP19,FraigniaudMORT19,NaorPY20}.
In this work, we add the quantum aspect, while considering only a single interaction, and only in the prescribed order: Merlin sends a proof to Arthur, and then there is no more interaction between them.

\subsection{Our Results} 
We carry on the recent trend of research consisting of investigating the power of quantum resources in the context of distributed network computing (cf., e.g., \cite{ElkinKNP14,GallM18,Izumi+PODC19,LeGall+STACS19,Izumi+STACS20,GavoilleKM09}), by designing a distributed Quantum Merlin-Arthur (\dQMA) protocol for distant \EQ, using compact certificates and small messages. 
While we use the \dQMA{} terminology in order to be consistent with prior work, we emphasize that the structure of the discussed protocols is rather simple: each node is given a quantum state as a certificate, the nodes exchange these states, perform a local computation, and finally accept or reject.

Our main result is the following. A collection of $n$-bit strings $x_1,\dots,x_t$ are stored at $t$ terminal nodes $u_1,\dots,u_t$ in a network $G=(V,E)$, where node~$u_i$ stores~$x_i$.  We denote $\EQ^t_n$ the problem of checking the equality $x_1=\dots=x_t$ between the $t$~strings.  
Let us define the \emph{radius} of a given instance of  $\EQ^t_n$ as $r=\min_i\max_{j} \mathsf{dist}_G(u_i,u_j)$, where $\mathsf{dist}_G$ denotes the distance in the (unweighted) graph~$G$. Our main result is the design of a \dQMA\/ protocol  for  $\EQ^t_n$, using small certificate.
This can be summarized by the following informal statement (the formal statement is in Section \ref{sec:general-graphs}): 

\begin{result}
	There is a distributed Quantum Merlin-Arthur (\dQMA) protocol for certifying equality  between $t$~binary strings ($\EQ^t_n$) of length~$n$, and located at a radius-$r$ set of $t$~terminals, in a single round of communication between neighboring nodes using certificates of size $O(t r^2 \log n)$ qubits, and messages of size 
	$O(t r^2 \log (n+r))$ qubits. 
\end{result}

It is worth mentioning that, although the dependence in~$r$ and $t$ is polynomial, the dependence in the actual size~$n$ of the instance remains logarithmic, which is our main concern. Indeed, for applications such as the aforementioned distributed data storage motivating the distant $\EQ^t_n$ problem, it is expected that both the number $t$ of replicas, and the maximum distance between the nodes storing these replicas are of several orders of magnitude smaller than the size~$n$ of the stored replicated data. 

It is also important to note that our protocol satisfies the basic requirement of \emph{reusability}, as one aims for protocols enabling   regular and frequent verifications that the data are not corrupted. Specifically, the quantum operations performed on the certificates during the local verification phase operated between neighboring nodes preserve the quantum nature of these certificates. That is, if $\EQ^t_n$ is satisfied, i.e., if all the replicas $x_i$'s are equal, then, up to an elementary local relocation of the quantum certificates, these certificates are available for a next test. If $\EQ^t_n$ is not satisfied, i.e., if there exists a pair of replicas $x_i\neq x_j$, then the certificates do not need to be preserved as this scenario corresponds to the case where the correctness of the data structure is violated, requiring the activation of recovery procedures for fixing the bug, and reassigning certificates to the nodes. 

Our quantum protocol is based on the SWAP test~\cite{BCWW01PRL}, which is a basic tool in the theory of quantum computation and quantum information. 
This test allows to check if a quantum state is symmetric, and has several applications, such as 
estimating the inner product of two states 
(e.g.,~\cite{BCWW01PRL,BT12CC,Yao03STOC}), 
checking whether a given state (or a reduced state of it) is pure or entangled with 
the environment system 
(e.g.,~\cite{ABDFS09ToC,KMY09CJTCS,HM13JACM,KLGN15SICOMP}), 
and more. 
In this paper, we use the SWAP test in 
yet another way: {\em for checking if two of the reduced states of a given state are close}. 
A similar use was done by Rosgen~\cite{Rosgen08} in a different context --- transforming quantum circuits to shallow ones in a hardness reduction proof.

Finally, observe that our logarithmic upper bound for $\dQMA$ protocols is in contrast to the linear lower bound that can be shown for classical $\dMA$ protocols even for $t=2$ on a path of 4 nodes and even for the case where communication between the neighboring nodes is extended to multiple rounds (see precise statement and proof in Section~\ref{sec:classicLWB}). Our results thus show that quantum certification mechanism can provide an exponential advantage over classical certification mechanisms.

\subsection{Related Work} 

The concept of distributed proofs is a part of the framework of distributed network computing since the early works on fault-tolerance (see, e.g., \cite{AfekKY97,AwerbuchPV91,ItkisL94}). 
Proof-labeling schemes were introduced in~\cite{KormanKP10}, and variants have been studied in~\cite{GoosS16,FraigniaudKP13}. Randomized proof-labeling schemes have been studied in~\cite{FraigniaudPP19}. Extensions of distributed proofs to a hierarchy of decision mechanisms have been studied in~\cite{FeuilloleyFH16} and~\cite{BalliuDFO18}. Frameworks like cloud computing recently enabled envisioning systems in which the nodes of the network could interact with a third party, leading to the concept of \emph{distributed interactive proofs}~\cite{KolOS18}. There, each node can interact with an  \emph{oracle} who has a complete view of the system, is computationally unbounded, but is not trustable. For instance, in Arthur-Merlin (\dAM) protocols, the nodes start by querying the oracle Merlin, which provides them with answers in their certificates. There is a simple classical compact \dAM\/ protocol for distant $\EQ$, where the two players stand at the extremities of a  path (see Section~\ref{sec:overview}). 
We refer to~\cite{CrescenziFP19,FraigniaudMORT19,NaorPY20} for recent developments in the framework of distributed interactive proofs. 
While distributed Arthur-Merlin protocols and their extensions provide an appealing theoretical framework for studying the power of interactive proofs in the distributed setting, the practical implementation of such protocols remains questionable, since 
they all require the existence of a know-all oracle, Merlin, and it is unclear if a Cloud could play this role.
On the other hand, in $\dMA$ and $\dQMA$ protocols, interaction with an external party is not required, but only a one-time assignment of certificates is needed, which are then reusable for regular verification.
As in the classical proof-labeling schemes setting, these certificates can  actually be \emph{created} by the nodes themselves during a pre-processing phase, making the reliance on a know-all oracle unnecessary.

After a few early works~\cite{Ben-Or+STOC05, ElkinKNP14,GavoilleKM09,Tani+12} that shed light on the potential and limitations of quantum distributed computing (see also \cite{Arfaoui+14,Broadbent+08,Denchev+08} for general discussions), evidence of the advantage of quantum distributed computing over classical distributed computing have been obtained recently for three fundamental models of (synchronous fault-free) distributed network computing:
the \textsf{CONGEST} model \cite{Izumi+STACS20,GallM18},
the \textsf{CONGEST-CLIQUE} model \cite{Izumi+PODC19}
and the \textsf{LOCAL} model \cite{LeGall+STACS19}.
The present paper adds to this list another important task for which quantum distributed computing significantly outperforms classical distributed computing, namely, distributed certification. 

Note that while this paper is the first to study quantum Merlin-Arthur protocols in a distributed computing framework, there are a number of prior works studying them 
in communication complexity \cite{RS04CCC,Kla11CCC,KP14MFCS,BGK15MFCS}. 
In particular, quantum Merlin-Arthur protocols are shown to 
improve some computational measure 
(say, the total length of the messages from the prover to Alice, and of the
messages between Alice and Bob) 
exponentially compared to Merlin-Arthur protocols where 
the messages from the prover are classical \cite{RS04CCC,KP14MFCS}. 

The question of computing functions on inputs that are given to graph nodes was also studied in the context of communication complexity. 
The equality function was studied for the case where all nodes have inputs~\cite{AlonES17}. Other works considered a setting similar to ours, i.e., where only some nodes have inputs~\cite{ChattopadhyayRR14,ChattopadhyayR15}, but did not study the equality problem.


\section{Model and Definitions}\label{sec:prelim}

\subparagraph*{Distributed verification on graphs.}

Let $t\geq 2$, and let $f\colon (\{0,1\}^n)^t \to\{0,1\}$ be a function. The aim of the nodes is to collectively decide whether $f(x_1,\dots,x_t)=1$ or not, where $x_1,\dots,x_t$ are assigned to $t$ nodes of a graph. Specifically, an instance of the problem~$f$ is a $t$-tuple $(x_1,\dots,x_t)\in \{0,1\}^n\times\dots\times\{0,1\}^n$, a connected graph $G=(V,E)$, and an ordered sequence $v_1,\dots,v_t$ of distinct nodes of~$G$. The node~$v_i$ is given $x_i$ as input, for $i=1,\dots,t$. All the other nodes receive no inputs. 
We consider distributed Merlin-Arthur ($\dMA$) protocols for deciding whether  $f(x_1,\dots,x_t)=1$, in which a non-trustable \emph{prover} (Merlin) assigns (or ``sends'') \emph{certificates} to the nodes, and then the nodes (Arthur) perform a 1-round randomized verification algorithm.
The verification algorithm consists of each node simultaneously sending messages to all its immediate neighbors, receiving messages from them, then performing a local computation, and finally accepting or rejecting locally.\footnote{We can naturally extend this definition to define $\dMA$ protocols with $\mu$ rounds of communication among neighbors, for any integer $\mu\ge 1$. In this paper, however, we focus on the case $\mu=1$ since all the protocols we design use only 1-round verification algorithms. The only exception is Section~\ref{sec:classicLWB}, where we show classical lower bounds that hold even for $\mu>1$.} We say that a $\dMA$ protocol has \emph{completeness}~$a$ and \emph{soundness}~$b$ for a function~$f$ if the following holds for every $(x_1,\dots,x_t)\in \{0,1\}^n\times\dots\times\{0,1\}^n$, every connected graph~$G$, and every ordered sequence $v_1,\dots,v_t$ of distinct nodes in~$G$:

\begin{description}
	\item[{(completeness)}] if $f(x_1,\dots,x_t)=1$, then  the prover can assign certificates to the nodes such that $\Pr[\mbox{all nodes accept}]\geq a$; 
	\item[{(soundness)}] if $f(x_1,\dots,x_t)=0$, then, for every certificate assignment  by the prover, $\Pr[\mbox{all nodes accept}]\leq b$.  
\end{description}

The completeness condition guarantees that, when the system is in a ``legal'' state (specified by $f(x_1,\dots,x_t)=1$), with probability at least~$a$ all nodes accept. The soundness condition guarantees that, when the system is in an ``illegal'' state (specified by $f(x_1,\dots,x_t)=0$), with probability at least~$1-b$ at least one node rejects. 
The value $b$ represents the error probability of the protocol on an illegal instance, and thus we sometimes refer to it as the \emph{soundness error}.
A node detecting illegality of the state can raise an alarm, or launch a recovery procedure. Protocols with completeness~1 are called 1-sided protocols, or protocols with perfect completeness.
Similarly to prior works on distributed verification, the certificate size of the protocol 
is measured as the maximum size (over all the nodes of the network) of the certificate sent by the prover to one of the nodes, and the message size of the protocol is measured as the maximum size (over all pairs of adjacent nodes) of the message exchanged between two adjacent nodes.
Specifically, we will consider the multi-party version of the equality function, 
$\EQ^{t}_{n}$, which is the boolean-valued function from 
$(\{0,1\}^n)^t$ such that 
$
\EQ^{t}_n(x_1,\ldots,x_t)=1 \iff x_1=\cdots=x_t.
$

In this work, we extend the framework of \dMA~protocols, to consider also cases where the certificates given to the nodes can contain qubits (although they may also contain classical bits) and the nodes can exchange messages consisting of qubits. These will be called \emph{distributed Quantum Merlin-Arthur} (\dQMA) protocols. More precisely, in a \dQMA~protocol for a function $f$, a non-trustable prover first sends a certificate to each node, which consists of a quantum state and classical bits; the quantum states may be entangled, even though all our quantum protocols do not require any prior entanglement, nor any shared classical random bits. Then the nodes perform a 1-round quantum verification algorithm, where each node simultaneously sends a quantum message to all its immediate neighbors, receives quantum messages from them, then performs a local computation, and finally accepts or rejects locally.
Note that, as opposed to the classical setting, we cannot assume that a node simply broadcasts its certificate to all its neighbors, as quantum states cannot be duplicated. However, a node can still send copies of the classical parts of the certificate.
We define completeness and soundness of \dQMA\/ protocols as for \dMA\/ protocols.

\subparagraph*{Remark.}

A special case of interest is when the graph $G$ is a path $v_0,\ldots,v_r$, $r\geq 1$, 
where the left-end node $v_0$ has an $n$-bit string~$x$ as input, the right-end node $v_r$ has an $n$-bit string~$y$ as input, 
and the intermediate nodes $v_1,\ldots,v_{r-1}$ have no inputs. That is, $t=2$. Given a function $f\colon \{0,1\}^n\times \{0,1\}^n\to\{0,1\}$, the aim of the nodes is to collectively decide whether $f(x,y)=1$ or not. This setting is very much related to communication complexity.

\subparagraph*{Classical two-party communication complexity.}

We refer to~\cite{KN97book} for the basic concepts of two-party communication complexity. In this paper we will only consider two-party one-way communication complexity. In this model two parties, denoted Alice and Bob, each receives an input $x\in \{0,1\}^n$ and~$y\in \{0,1\}^n$, respectively. The goal is for Bob to output the value $f(x,y)$ for some known Boolean function $f\colon \{0,1\}^n\times \{0,1\}^n\rightarrow \{0,1\}$. Only Alice can send a message to Bob. The one-way two-sided-error communication complexity of $f$ is the minimum number of bits that have to be sent on the worst input in a protocol that outputs the correct answer with probability at least 2/3. The one-way one-sided-error communication complexity of $f$ is the minimum number of bits that have to be sent on the worst input in a protocol that outputs the correct answer with probability~1 on any $1$-input, and outputs the correct answer with probability at least~2/3 on any $0$-input.

We shall especially consider the following two functions. The equality function $\EQ_n$ is defined as $\EQ_n(x,y)=1$ when $x=y$ and $\EQ_n(x,y)=0$ otherwise, for any $x,y\in\{0,1\}^n$. Its one-way one-sided-error communication complexity is $O(\log n)$ --- see, e.g.,~\cite{KN97book}. For any integer $d\ge 0$, the Hamming distance function $\HAM_{n}^d$ is defined as follows: for any $x,y\in \{0,1\}^n$, $\HAM_{n}^d(x,y)=1$ if the Hamming distance between $x$ and $y$ is at most $d$, and $\HAM_n^d(x,y)=0$ otherwise.  It is known~\cite{Yao03STOC} that, for $d$ constant, the one-way two-sided-error communication complexity of $\HAM_{n}^d$ is $O(\log n)$.

For any Boolean function $f: \{0,1\}^n\times \{0,1\}^n\rightarrow \{0,1\}$, 
a set $S\subseteq \{0,1\}^n\times \{0,1\}^n$ is a \emph{$1$-fooling set} for $f$ 
if, on the one hand, for every $(x,y)\in S$, $f(x,y)=1$,  and, on the other hand, for every two pairs $(x_1,y_1)\neq (x_2,y_2)$ in $S\times S$, $f(x_1,y_2)=0$ or $f(x_2,y_1)=0$. 

\subparagraph{Quantum two-party communication complexity.}

We assume the reader is familiar with the basics of quantum computation, 
in particular the notion of qubits, Dirac notation such as $|\psi\rangle$ and $\langle\psi|:=(|\psi\rangle)^\dagger$, and the quantum circuit model (see Sections 2 and 4 in Ref.~\cite{NC00book}, for instance). In Appendix~\ref{sec:QuantFund} we present more advanced
concepts such as mixed states that will be used in some of our proofs.

Quantum two-party communication complexity, first introduced by Yao \cite{YaoFOCS93}, is defined similarly to the classical version. The only difference is that the players are allowed to exchange qubits instead of bits (the cost of a quantum protocol is the number of qubits sent by the protocol). 
Note that since quantum protocols can trivially simulate classical protocols, the quantum communication complexity of a function is never larger than its classical communication complexity. 
More precisely, an $m$-qubit one-way quantum protocol $\pi$ for the function $f$ can be described in its most general form as follows. Alice prepares an $m$-qubit (pure) quantum state $\ket{h_x}$ and sends it to Bob.\footnote{Without loss of generality, we assume that Alice does not use any mixed state (i.e., a probability distribution on pure states) in her message, as she can simulate it using a pure state called the \emph{purification}~\cite{NC00book} whose length is at most twice the one of the mixed state.} Bob then makes a measurement on the state $\ket{h_x}$, which gives an outcome $b\in\{0,1\}$. Finally, Bob outputs~$b$. Since Bob's measurement in the above description depends only on his input~$y$, it can be mathematically described, for each $y\in\{0,1\}^n$, by two positive semi-definite matrices $M_{y,0}$ and $M_{y,1}$ such that $M_{y,0}+M_{y,1}=I$. This pair $\{M_{y,0},M_{y,1}\}$ is called a POVM measurement (POVM measurements are the most general form of measurements allowed by quantum mechanics). 
If $\ket{h_x}$ is measured by the POVM $\{M_{y,0},M_{y,1}\}$, the probability that $b=0$ is $\tr(M_{y,0}(|h_x\rangle\langle h_x|))$, while the probability that $b=1$ is $\tr(M_{y,1}(|h_x\rangle\langle h_x|))$.

\section{General Overview of our Techniques} 
\label{sec:overview}

Let us provide an intuition of our protocol in the case of $\EQ^2_n$ over a path $v_0,\ldots,v_r$ of length $r\ge 1$ in which the terminals are the two nodes $v_0$ and~$v_r$ (that we rename Alice and Bob, for convenience). Let us call $x$ and $y$ the $n$-bit strings owned by Alice and Bob, respectively. There is a simple classical protocol for distant equality in a somewhat similar setting, where the verifier (Arthur, consisting of all the graph nodes) can send random bits to the prover (Merlin) before receiving the certificates; this is called a $\dAM$ protocol. In this protocol, Alice picks a hash function~$h$ at random in an appropriate family of hash functions (i.e., a family such that both $h$ and $h(x)$ can be encoded using  $O(\log n)$ bits and such that the probability that $h(x)\neq h(y)$ is high when $x\neq y$). Merlin provides every node with the certificate $(h,h(x))$, each node checks it received the same certificates as its neighbors, and Bob additionally checks whether $h(x)=h(y)$. Obviously, one cannot switch the order of Arthur and Merlin, as letting Merlin choose the hash function would enable him to fool Arthur on illegal instances by picking $h$ that hashes identically the distinct input strings $x$ and~$y$. The main idea of our $\dQMA$ protocol is to ask Merlin to provide the nodes with a quantum certificate consisting of the \emph{quantum superposition} of all the possible hashes. 

Entering slightly more into the details, for any $x\in \{0,1\}^n$ we consider the  
{\em quantum fingerprint} $|h_x\rangle=\frac{1}{\sqrt{K}}\sum_{h}|h\rangle|h(x)\rangle$, where the sum is over all the hash functions, and $\frac{1}{\sqrt{K}}$ is the normalization factor of the quantum state. By using the same family of hash functions as in the aforementioned $\dAM$ protocol, these fingerprints can be constructed in such a way that their length is $O(\log n)$ qubits, and $|h_x\rangle$ and $|h_y\rangle$ are very far (more precisely, almost orthogonal) when $x\neq y$. Checking whether the two quantum fingerprints $|h_x\rangle$ and $|h_y\rangle$ 
are either equal or far apart can be achieved by a quantum test called the SWAP test~\cite{BCWW01PRL}.
Formally, the probability that the SWAP test accepts  is $1/2+|\langle h_x|h_y\rangle|^2/2$, where $\langle h_x|h_y\rangle$ denotes the inner product between the two quantum states $|h_x\rangle$ and $|h_y\rangle$.

Let us now describe the outline of our $\dQMA$ protocol. In the protocol each intermediate node $v_1,\ldots,v_{r-1}$ expects to receive the quantum fingerprint $|h_x\rangle$. Alice, who does not receive any certificate, creates by herself the fingerprint $|h_x\rangle$, which depends only on $x$. Similarly, Bob creates by himself the fingerprint $|h_y\rangle$. The checking procedure simply checks whether all these $(r+1)$ fingerprints are equal. This is done by applying the SWAP test to check whether the fingerprints owned by adjacent nodes are equal or not. There are however a few subtleties. In particular, since our analysis crucially requires that the SWAP tests do not overlap, for each node we need to decide whether it will perform the SWAP test with its right neighbor or its left neighbor. We do it in a randomized way and deal carefully with the conflicting choices that appear. For the case $x=y$ all the SWAP tests then succeed with probability 1 and thus all the nodes accept.

For the case $x\neq y$,  let us provide some intuition about why a prover cannot fool the nodes for convincing them to all accept. To simplify the description we assume below that $|h_x\rangle$ and $|h_y\rangle$ are orthogonal (instead of only almost orthogonal).
If the prover was forced to send certificates restricted to {\em product states} of the form $|g_1\rangle\otimes|g_2\rangle\otimes \cdots |g_{r-1}\rangle$ 
where $|g_j\rangle$ is the state to the $j$th node, then a fairly straightforward argument would guarantee that, with large probability, at least one node rejects. 
Indeed, under the product states restriction, intuitively 
the best  strategy for the prover to cheat is to send 
states ``intermediate'' between $|h_x\rangle$ and $|h_y\rangle$,
namely, to send the state $|g_j\rangle=\cos(\pi j/2r)|h_x\rangle+\sin(\pi j/2r)|h_y\rangle$ to node $v_j$ for each $j\in\{1,\ldots,r-1\}$. 
Then, the probability that all nodes accept when performing the SWAP tests would be roughly 
$\prod_{j=1}^{r-1}(1/2+|\langle g_j|g_{j+1}\rangle|^2/2)=1-\Omega(1/r)$.
The cheating prover could then be caught with probability $\Omega(1/r)$, 
and this probability can be amplified to $\Omega(1)$ by asking the prover to send several copies of the certificates (amplification is possible since our protocol has perfect completeness).

The formal analysis of the protocol however faces several difficulties, which are mostly due to the nature of quantum computation, and are especially challenging to handle in the framework of distributed computation. For instance, quantum states cannot be duplicated (the ``no-cloning Theorem''), which implies that a same quantum state cannot be used for parallel tests. Additionally, even sequential tests face the difficulty that the first test may collapse the quantum state, making the second test impossible to perform (or at least significantly complicating the analysis of the second test). Thus node $v_i$ cannot perform the SWAP test with its two neighbors $v_{i-1}$ and $v_{i+1}$ simultaneously and (as already mentioned) we have to design carefully the protocol so that the SWAP tests do not overlap. A second, and much more problematic issue is that the non-trustable Merlin can send arbitrary certificates to the nodes for fooling them. In particular it is not restricted to send certificates that are product states. A priori, it may seem that the SWAP test is not strong enough to handle fooling  strategies beyond product states. In this work we show that the SWAP test can actually detect such fooling strategies. 

Specifically, our approach consists in considering the so-called  \emph{reduced states}, 
and to establish the following property of the SWAP test (cf. Lemma~\ref{lem:swaptest_close} in Section~\ref{sec:QDPP}). If the SWAP test accepts with high probability when applied on the part of any two adjacent nodes 
in a (possibly non-product) global quantum state resulting from the certificates, 
then the two reduced states of that part (which is a bipartite state) must be close.
As the two states $|h_x\rangle$ and $|h_y\rangle$ are very far apart when $x\neq y$, 
we can thus use this result to show that there is a good probability that the SWAP test 
rejects at some node.  Moreover, using reduced states allows us to overcome other technical difficulties in the analysis of the (non-overlapping) SWAP tests we consider. Indeed, some form of average-case success probability of all the SWAP tests can be considered, instead of having to argue about the probability that all the SWAP tests individually accept.


\section{Quantum Distributed Proofs on Paths}
\label{sec:QDPP}

In this section we restrict ourselves to the case of a path $v_0,\ldots,v_r$ of length $r\geq 1$, in which only the two extremities $v_0$ and $v_r$ are given inputs. This framework allows us to elaborate our main technique, that will be extended to arbitrary graphs in Section \ref{sec:general-graphs}. Let $x\in\{0,1\}^n$ be the input to $v_0$, and $y\in\{0,1\}^n$ be the input to $v_r$. Our goal is to design a $\dQMA$ protocol to decide whether $f(x,y)=1$ or not, for some given Boolean function on $\{0,1\}^n\times \{0,1\}^n$. 

We show the following general theorem
that converts a one-way quantum communication complexity protocol into a quantum Merlin-Arthur protocol for the corresponding long-distance problem on the path. This theorem applies not only to one-sided-error protocols, but also to the two-sided-error case (with a logarithmic additional factor in the complexity).

\begin{theorem}
	\label{th:general-one-two-sided} 
	Let $f\colon\{0,1\}^n\times\{0,1\}^n\to \{0,1\}$ be a Boolean function.
	\begin{itemize}
		\item 
		If $f$ has a quantum one-way one-sided-error communication protocol transmitting at most $q$ qubits, then there exists a 1-sided distributed quantum Merlin-Arthur protocol for $f$ on the path of length~$r$, with soundness~$1/3$, using certificates of size $O(r^2q)$ qubits, and exchanging messages of length $O(r^2(q+\log r))$ qubits. 
		
		\item
		If $f$ has a quantum one-way two-sided-error communication protocol transmitting at most $q$ qubits, then, for any constant $c$, there exists a distributed quantum Merlin-Arthur protocol for $f$ on the path of length~$r$  with completeness $1-1/n^c$, soundness~$1/3$, using certificates of size $O(r^2q \log (n+r))$ qubits, 
		and exchanging messages of length $O(r^2q \log (n+r))$ qubits. 
	\end{itemize}
\end{theorem}

Using known results (cf. Section~\ref{sec:prelim}) about one-way communication complexities of $\EQ_n$ and $\HAM_{n}^d$, the following  two results are direct applications of Theorem~\ref{th:general-one-two-sided}.

\begin{corollary}
	There exists a one-sided quantum Merlin-Arthur protocol for $\EQ_n$ in the path of length~$r$  with soundness $1/3$, using certificates of size $O(r^2\log n)$ qubits, and exchanging messages of length $O(r^2 \log (n+ r))$ qubits.\footnote{Here we are using the fact that $\log n+\log r$ is of the same order as $\log(n+r)$ for conciseness.}
\end{corollary}

\begin{corollary}
	For any $c>0$ and $d>0$, there exists a quantum Merlin-Arthur protocol for $\HAM_{n}^d$ in the path of length~$r$  with completeness $1-1/n^c$, soundness $1/3$, using certificates of length $O(r^2(\log n)\log (n+r))$ qubits, 
	and exchanging messages of length $O(r^2(\log n)\log (n+r))$ qubits.  
\end{corollary}

The rest of this section is dedicated to proving  Theorem~\ref{th:general-one-two-sided}. 
Let us first give an overview of the proof.
In our $\dQMA$ protocol in the path, the verification algorithm performed by the nodes on the line is merely a simulation of  a two-party one-way quantum communication complexity protocol $\pi$ between Alice and Bob for the function $f(x,y)$, with the help of certificates provided by the prover. Specifically, every intermediate node $v_1,\dots,v_{r-1}$ expects to receive the quantum state sent by Alice to Bob in $\pi$, as certificate. Let us denote by $|h_x\rangle$ this state, which depends on~$x$. The right-end node $v_r$ simulates the two-party protocol $\pi$ using $|h_x\rangle$ received from the left neighbor $v_{r-1}$, and applying Bob's measurement (i.e., the POVM measurement). If $f(x,y)=1$, the prover honestly sends the desired state, and $v_r$ accepts as it does receive $|h_x\rangle$. However, if $f(x,y)=0$, then the malicious prover does not necessarily send a desired state. To catch the potentially malicious behavior of the prover on ``illegal'' instances (i.e., those for which $f(x,y)=0$), each intermediate node checks whether its local proof is ``close to'' the one of its right neighbor. This is performed by an application of the SWAP test. 

Section \ref{sec:protocol} below describes in more detail how to construct the distributed quantum Merlin-Arthur protocol, denoted ${\cal P}_\pi$, from an arbitrary one-way quantum communication protocol~$\pi$ for the function~$f$. 
Section~\ref{sec:analysis} analyzes the completeness and the soundness of the protocol ${\cal P}_\pi$. Finally, Section \ref{sec:proof} shows how to reduce the soundness error using ``parallel repetitions'' and how to apply this analysis 
to prove Theorem~\ref{th:general-one-two-sided}.

\subsection{A $\dQMA$ Protocol for the Path}\label{sec:protocol}

Let $\varepsilon\ge 0$ be a constant, which will be fixed small enough later in the proof. Let $\pi$ be a quantum one-way communication protocol for $f$ transmitting at most $q$~qubits, such that, for every input pair $(x,y)$, if $f(x,y)=1$ then $\pi$ outputs~$1$ with probability at least $1-\varepsilon$, and if $f(x,y)=0$ then $\pi$ outputs~$0$ with probability at least~$2/3$.
Let $|h_x\rangle$ be the $q$-qubit (pure) state sent from Alice to Bob, 
and let $\{M_{y,1},M_{y,0}\}$ be the POVM measurement 
performed by Bob on $|h_x\rangle$, where $M_{y,1}$ corresponds to 
the measurement result~$1$ (accept) and 
$M_{y,0}$ to the measurement result~$0$ (reject). 
Our quantum Merlin-Arthur protocol ${\cal P}_\pi$ is as follows.

\begin{center}
	\fbox{
		\begin{minipage}{13cm}
			{\bf Protocol ${\cal P}_\pi$ for function $f$ on input pair $(x,y)$ in path $v_0,\dots,v_r$:}
			\begin{enumerate}
				\item 
				If $f(x,y)=1$ then the prover sends the quantum register $R_j$ 
				that has the state $|h_x\rangle$ (or $|h_x\rangle\langle h_x|$ as the mixed state representation) as certificate 
				to each of the intermediate nodes $v_j$, $j\in\{1,\ldots,r-1\}$. 
				\item The left-end node $v_0$ prepares the state $\rho_0=|h_x\rangle\langle h_x|$ in quantum register $R_0$.
				\item 
				For every $j=0,\ldots,r-1$, the node $v_j$ chooses a bit $b_j$ uniformly at random, and sends its quantum register $R_j$ to the right neighbor $v_{j+1}$ whenever $b_j=0$. 
				\item
				For every $j=1,\ldots,r-1$, if $v_j$ receives a quantum register from its left neighbor $v_{j-1}$, and if $b_{j}=1$, then $v_j$ performs the SWAP test on the registers 
				$(R_{j-1},R_j)$, and accepts or rejects accordingly; Otherwise, $v_j$ accepts. 
				\item
				If the right-end node $v_r$ receives a quantum register $R_{r-1}$ from its left neighbor, then $v_r$ performs the POVM measurement $\{M_{y,1},M_{y,0}\}$ 
				corresponding to $\pi$ applied to the state in $R_{r-1}$, and accepts or rejects accordingly; Otherwise, $v_r$ accepts.
			\end{enumerate}
		\end{minipage}
	}
\end{center}

In the above protocol ${\cal P}_\pi$, the size of the quantum certificate that each node receives from the prover is at most~$q$,
and the length of the quantum message that each node sends to the neighbor is also at most $q$.  
In the next subsection, we prove that the above protocol  has completeness $1-\varepsilon /2$ and soundness $1-1/42r^2$.

\subsection{Analysis of Protocol ${\cal P}_\pi$}\label{sec:analysis}

For the analysis, we recall the SWAP test. 
The test is a protocol with a given input state on ${\cal H}={\cal H}_1\otimes {\cal H}_2$,  where ${\cal H}_1$ and ${\cal H}_2$ are complex Euclidian spaces. Here, we consider ${\cal H}_1$ and ${\cal H}_2$ as quantum registers $R_1$ and $R_2$.

\begin{center}
	\fbox{
		\begin{minipage}{11cm}
			{\bf SWAP test} on a pure state $|\psi\rangle$ on ${\cal H}$, which is given in registers $(R_1,R_2)$.
			\begin{enumerate}
				\item Prepare the single-qubit state $|+\rangle=\frac{1}{\sqrt{2}}(|0\rangle+|1\rangle)$ in register $R_0$.
				\item If the content of $R_0$ is $1$, then apply the swap operator $S$ on the state $|\psi\rangle$ in registers $(R_1,R_2)$, where $S$ is defined by $S(|j_1\rangle|j_2\rangle)=|j_2\rangle|j_1\rangle$ (namely, $S$ swaps register $R_1$ and register $R_2$). 
				\item Apply the Hadamard operator $H=\frac{1}{\sqrt{2}}\begin{psmallmatrix}1 & 1\\1 & -1\end{psmallmatrix}$ on the state in register $R_0$, and measure the content in the standard basis. Accept if the content is $0$, and reject otherwise. 
			\end{enumerate}
		\end{minipage}
	}
\end{center}

\subparagraph{Completeness.}

The following lemma is a direct consequence of the definition of the SWAP test. 
Here, ${\cal H}_S$ is the symmetric subspace 
in ${\cal H}$ (namely, the subspace spanned by the states invariant 
by the swap operator $S$, or equivalently, the eigenstates of $S$ with eigenvalue $1$),
and ${\cal H}_A$ is the anti-symmetric subspace in ${\cal H}$ 
(namely, the subspace spanned by the eigenstates of $S$ with eigenvalue $-1$). 
Note that any state in ${\cal H}$ is represented as the superposition 
of a state in ${\cal H}_S$ (symmetric state) and a state in ${\cal H}_A$ 
(anti-symmetric state)
as the swap operator $S$ is a Hermitian matrix that only has $+1$ and $-1$ eigenvalues.

\begin{lemma}\label{lem:swaptest}
	Assume that $|\psi\rangle=\alpha|\psi_S\rangle+\beta|\psi_A\rangle$ 
	where $|\psi_S\rangle \in {\cal H}_S$ and $|\psi_A\rangle \in {\cal H}_A$.
	Then, the SWAP test on input $|\psi\rangle$ accepts with probability $|\alpha|^2$. 
\end{lemma}

\begin{proof}
Noting that $S|\psi_S\rangle=|\psi_S\rangle$ and $S|\psi_A\rangle=-|\psi_A\rangle$, the state after Step 2 is 
\begin{align*}
	\frac{1}{\sqrt{2}}(|0\rangle|\psi\rangle+|1\rangle S|\psi\rangle)
	&= \frac{1}{\sqrt{2}}(\alpha|0\rangle|\psi_S\rangle+\beta|0\rangle|\psi_A\rangle+\alpha|0\rangle|\psi_S\rangle-\beta|1\rangle|\psi_A\rangle)\\
	&= \frac{1}{\sqrt{2}}[|0\rangle(\alpha|\psi_S\rangle+\beta|\psi_A\rangle)+|1\rangle(\alpha|\psi_S\rangle-\beta|\psi_A\rangle)].
\end{align*}
The final state obtained in Step 3 is 
\begin{align*}
	\lefteqn{\frac{1}{\sqrt{2}}[(H|0\rangle)(\alpha|\psi_S\rangle+\beta|\psi_A\rangle)+(H|1\rangle)(\alpha|\psi_S\rangle-\beta|\psi_A\rangle)]} \\
	&= \frac{1}{2}[(|0\rangle+|1\rangle)(\alpha|\psi_S\rangle+\beta|\psi_A\rangle)+(|0\rangle-|1\rangle)(\alpha|\psi_S\rangle-\beta|\psi_A\rangle)]\\
	&= \alpha|0\rangle|\psi_S\rangle+\beta|1\rangle|\psi_A\rangle.
\end{align*} 
Thus, the probability that $0$ is measured on $R_0$ (and thus is accepted) is $|\alpha|^2$.
\end{proof}

For the completeness, assume $f(x,y)=1$. The prover then sends $|h_x\rangle$ to all the intermediate nodes.
Then, all the nodes except the right-end node have $|h_x\rangle$. By Lemma~\ref{lem:swaptest}, all the SWAP tests 
done in Step 4 are accepted with probability $1$ 
(note that $|h_x\rangle\otimes|h_x\rangle$ is a symmetric state). 
Furthermore, the right-end node accepts with probability 
at least $(1-\varepsilon)/2+1/2=1-\varepsilon/2$ 
as $v_r$ can receive $|h_x\rangle$ and simulate $\pi$ with probability $1/2$ 
and accepts otherwise in Step 5.

\subparagraph{Soundness.}

The following lemma presents a crucial property of the SWAP test: its applicability
for checking whether the two \emph{reduced} states are close. 
It is a trace-distance version of a lemma by Rosgen~\cite[Lemma~5.1]{Rosgen08}.
Here, a reduced state intuitively represents the local information on its own quantum system, by disregarding the outside systems. 
Note that the trace distance between two quantum states 
$\rho$ and $\sigma$ is characterized as
$
\mathsf{dist}(\rho,\sigma)=\max_{M} \tr(M(\rho-\sigma)),
$
where the maximization is taken over all positive semi-definite matrices $M$ 
such that $M\leq I$. 

\begin{lemma}\label{lem:swaptest_close}
	Let $\K\geq 1$, and assume that the SWAP test on input $\rho$ in the input register $(R_1,R_2)$ 
	accepts with probability $1-1/\K$.
	Then, $\mathsf{dist}(\rho_1, \rho_2 ) \leq 2/\sqrt{\K}+ 1/\K$, 
	where $\rho_j$ is the reduced state on $R_j$ of $\rho$.
	Moreover, if the SWAP test on input $\rho$ accepts with probability~1, 
	then $\rho_1=\rho_2$ (and hence $\mathsf{dist}(\rho_1, \rho_2)=0$).
\end{lemma}

\begin{proof}

First, we observe the second statement. By Lemma~\ref{lem:swaptest}, if $\rho$ includes some asymmetric state, the SWAP test rejects with a non-zero probability.
Hence, $\rho$ must consist of only symmetric states, which means that the two reduced states of $\rho$ coincides.  

In the remaining part, we prove the first statement. 
The mixed state $\rho$ can be represented as $\rho=\sum_j p_j |\psi_j\rangle\langle \psi_j|$ 
(which means that the state is in a (pure) state $|\psi_j\rangle$ with probability $p_j$).
Moreover, each $|\psi_j\rangle$ can be represented as a superposition of a symmetric state and an antisymetric state,
namely, $|\psi_j\rangle=\alpha_j|\psi_{j}^S\rangle+\beta_j|\psi_{j}^A\rangle$ with some symmetric state $|\psi_{j}^S\rangle$
and some antisymmetric state $|\psi_{j}^A\rangle$, where $|\alpha_j|^2+|\beta_j|^2=1$. 
Then, by Lemma~\ref{lem:swaptest} and the assumption, $\sum_j p_j |\alpha_j|^2\geq 1-1/\K$, and thus,
\begin{equation}\label{eq:beta}
	\sum_j p_j |\beta_j|^2\leq \frac{1}{\K}.
\end{equation}
The mixed state $\rho$ is furthermore rewritten as 
\begin{equation}\label{eq:rho}
	\rho=\sum_j p_j 
	\left( |\alpha_j|^2 |\psi_j^S\rangle\langle\psi_j^S|+\alpha_j\beta_j^* |\psi_j^S\rangle\langle \psi_j^A|
	+\alpha_j^*\beta |\psi_j^A\rangle\langle \psi_j^S|+|\beta|^2 |\psi_j^A\rangle\langle\psi_j^A|
	\right),
\end{equation}
and the reduced state $\rho_{3-i}=\tr_{i}(\rho)$ ($i=1,2$) 
on ${\cal H}_{3-i}$ (obtained by tracing out on ${\cal H}_i$) is
\begin{align*}
	\tr_{i}(\rho)=\sum_j p_j &
	\left( |\alpha_j|^2 \tr_i (|\psi_j^S\rangle\langle\psi_j^S|)
	+\alpha_j\beta_j^* \tr_i(|\psi_j^S\rangle\langle \psi_j^A|) \right.\\
	& \left. +\alpha_j^*\beta \tr_i(|\psi_j^A\rangle\langle \psi_j^S|)
	+|\beta_j|^2 \tr_i(|\psi_j^A\rangle\langle\psi_j^A|)
	\right).
\end{align*}
As $\tr_1(|\psi_j^S\rangle\langle\psi_j^S|)=\tr_2(|\psi_j^S\rangle\langle\psi_j^S|)$ from the definition of the symmetric state, 
\begin{align*}
	&\!\!\!\!\!\!\!\!  \tr_1(\rho)-\tr_2(\rho)\\
	&\!\!\!\! =\sum_j p_j 
	\left(
	\alpha_j\beta_j^* [\tr_1 (\rho_j^{sa})-\tr_2 (\rho_j^{sa})]
	+\alpha_j^*\beta_j [\tr_1 (\rho_j^{as})-\tr_2 (\rho_j^{as})]
	+|\beta_j|^2 [\tr_1 (\rho_j^{a}) - \tr_2(\rho_j^{a})]
	\right),
\end{align*}
where $\rho_j^{sa}=|\psi_j^S\rangle\langle\psi_j^A|$, $\rho_j^{as}=|\psi_j^A\rangle\langle\psi_j^S|$, and $\rho_j^{a}=|\psi_j^A\rangle\langle\psi_j^A|$. 
By the positive scalability and the triangle inequality of the trace norm, 
$\| \tr_1(\rho)-\tr_2(\rho) \|_{\tr}$ is at most
\begin{align*}
	\sum_j p_j
	&\left(
	|\alpha_j||\beta_j| \|\tr_1(\rho_j^{sa})-\tr_2(\rho_j^{sa})\|_{\tr}
	+ |\alpha_j||\beta_j| \|\tr_1(\rho_j^{as})-\tr_2(\rho_j^{as})\|_{\tr}\right.\\
	&   \left.  + |\beta_j|^2 \|\tr_1(\rho_j^{a})-\tr_2(\rho_j^{a})\|_{\tr}
	\right).
\end{align*}
As $\rho_j^a$ is a quantum state (with trace norm $1$), we have 
\[
\|\tr_1(\rho_j^{a})-\tr_2(\rho_j^{a})\|_{\tr} 
\leq \|\tr_1(\rho_j^{a})\|_{\tr}+\|\tr_2(\rho_j^{a})\|_{\tr}=1+1\leq 2.
\]
On the contrary, we notice that $\rho_j^{sa}$ (or $\rho_j^{as}$) is not a quantum state. 
However, by Lemma~\ref{lem:tracenorm} and the fact that the fidelity between two quantum states is at most $1$,  
\[
\| \tr_i (\rho_j^{sa}) \|_{\tr} 
= \| \tr_i (|\psi_j^S\rangle\langle\psi_j^A|) \|_{\tr} 
= F(\tr_{3-i}(|\psi_j^A\rangle\langle\psi_j^A|), \tr_{3-i}(|\psi_j^S\rangle\langle\psi_j^S|))\leq 1,
\]
and thus,
\[
\|\tr_1(\rho_j^{sa})-\tr_2(\rho_j^{sa})\|_{\tr}\leq 2.
\]
Similarly, we have $\|\tr_1(\rho_j^{as})-\tr_2(\rho_j^{as})\|_{\tr}\leq 2$. Therefore,
\begin{equation}\label{eq:last}
	\| \tr_1(\rho)-\tr_2(\rho) \|_{\tr}\leq \sum_j p_j 4|\alpha_j||\beta_j|+\sum_j p_j 2|\beta_j|^2.
\end{equation}
By Eq.~(\ref{eq:beta}), the second term of the right-hand side is at most $2/\K$. 
By the Cauchy-Schwarz inequality (and $|\alpha_j|\leq 1$), 
the first term is at most
\[
4 \sum_j \sqrt{p_j} \sqrt{p_j} |\beta_j|\leq 4\left(\sum_j p_j \right)^{1/2}\left(\sum_j p_j |\beta_j|^2 \right)^{1/2}\leq 4\sqrt{\frac{1}{\K}}. 
\]
This induces that 
\[
\mathsf{dist}(\rho_1,\rho_2) =\frac{1}{2}\| \tr_2(\rho)-\tr_1(\rho) \|_{\tr}\leq \frac{1}{\K}+2\sqrt{\frac{1}{\K}}.
\qedhere\]
\end{proof}

For the soundness, let $(x,y)$ be any pair such that $f(x,y)=0$.

\begin{lemma}\label{lem:soundness}
	For every $j\in\{1,\ldots,r\}$, let $F_{j}$ be the event that $v_j$ performs the local test (SWAP or POVM) in Protocol~${\cal P}_\pi$, and let $E_{j}$ be the event that this local test rejects. 
	Then we have 
	$
	\sum_{j=1}^{r}\Pr[E_{j}|F_{j}]\geq \frac{1}{21r}.
	$
\end{lemma}

\begin{proof}
Let $\alpha_j=\Pr[E_j|F_j]$. Then, for every $j\in\{1,\ldots,r\}$, $\Pr[\overline{E_{j} }|F_{j}]=1-\alpha_j$, 
where we note that the complementary event $\overline{E_{j}}$ for $j=1,\ldots,r-1$ 
is the event where the SWAP test on the two $q$-qubit states $\rho_{j-1}$ and $\rho_j$ accepts, and $\overline{E_r}$ represents the event that 
the result of the POVM measurement is $1$ (accept),  
which corresponds to the POVM element $M_{y,1}$. 
By Lemma \ref{lem:swaptest_close}, 
the trace distance $\mathsf{dist}$ between the reduced $q$-qubit states $\rho_{j-1}$ on $v_{j-1}$ and $\rho_{j}$ on $v_j$ is
\[
\mathsf{dist}(\rho_{j-1},\rho_{j})\leq 
\left\{
\begin{array}{ll}
	\frac{2}{\sqrt{1/\alpha_j}} + \frac{1}{1/\alpha_j} & \mbox{if $\alpha_j\neq 0$}\\
	0 & \mbox{otherwise}
\end{array}
\right. 
\] 
Thus $\mathsf{dist}(\rho_{j-1},\rho_{j})\leq 3\sqrt{\alpha_j}. $
Then, by the triangle inequality, 
\[
\mathsf{dist}(\rho_0,\rho_{r-1})\leq \sum_{j=1}^{r-1} \mathsf{dist}(\rho_{j-1},\rho_{j})
\leq 3 \sum_{j=1}^{r-1} \sqrt{\alpha_j}.
\]
For $\Pr[E_r|F_r]$, the soundness of $\pi$ promises that the probability that 
the test $\{M_{y,1},M_{y,0}\}$ rejects $\rho_0=|h_x\rangle\langle h_x|$ is at least $2/3$, 
i.e., $\tr(M_{y,0}\rho_0 ) \geq 2/3$. Hence, 
\[
\alpha_r=\Pr[E_r|F_r]=\tr(M_{y,0}\rho_{r-1} )\geq \frac{2}{3}-\mathsf{dist}(\rho_0,\rho_{r-1})\geq \frac{2}{3}-3 \sum_{j=1}^{r-1} \sqrt{\alpha_j},
\]
where the first inequality comes from the characterization of $\mathsf{dist}$ on 
indistinguishability of two states, that is,  $|\tr(M_{y,0}\rho_0)-\tr(M_{y,0}\rho_{r-1})|\leq \mathsf{dist}(\rho_0,\rho_{r-1})$. 
Thus, we have
\[
3\sum_{j=1}^{r} \sqrt{\alpha_j} \geq \alpha_r + 3 \sum_{j=1}^{r-1} \sqrt{\alpha_j} \geq \frac{2}{3}.
\]
By the Cauchy-Schwarz inequality,
\[
\sqrt{r}\sqrt{\sum_{j=1}^{r}\alpha_j}\geq \sum_{j=1}^{r} \sqrt{\alpha_j},
\]
and thus we have 
\[
\sum_{j=1}^{r} \alpha_j \geq \left(\frac{2}{9\sqrt{r}}\right)^2 \geq \frac{1}{21r}.
\qedhere\]	
\end{proof}

In Steps~4 and~5, node $v_j$, $1\leq j \leq r$, performs  
the local test (SWAP or POVM) with probability at least $1/4$ 
(more precisely, $v_j$, $1\leq j\leq r-1$, performs the SWAP test with probability $1/4$ and $v_r$ performs the POVM with probability $1/2$). 
It follows that, for every  $j\in\{1,\ldots,r\}$, 
the event $F_{j}$ occurs in at least $(1/4)\times 2^{r}$ outcomes  
of all the $2^{r}$ possible outcomes $b_0\cdots b_{r-1}$ 
that can be obtained in Step 3. 
Here, we consider any fixed outcomes $b_0\cdots b_{r-1}$ that induce $k$ events $F_{j_1},F_{j_2},\ldots,F_{j_k}$ with $k\neq 0$ where 
we note that $0\leq k \leq \lfloor r/2\rfloor$ in general.  
The probability that some node rejects in Steps~4 or~5 {\em under this outcome} is 
\[
\Pr[\vee_{i=1}^k  E_{j_i}|\wedge_{i=1}^{k} F_{j_i}]
\geq \frac{1}{\lfloor r/2\rfloor} \sum_{i=1}^k \Pr[E_{j_i}|\wedge_{i=1}^{k} F_{j_i}]
=    \frac{1}{\lfloor r/2\rfloor} \sum_{i=1}^k \Pr[E_{j_i}|F_{j_i}], 
\]
where the inequality comes from Lemma~\ref{lem:probability} 
in Appendix~\ref{sec:ProbFund}, 
and the equality comes from the fact that each of $F_{j_i}$ and $E_{j_i}$ 
is independent from all the other event $F_{j_{i'}}$ with $i'\neq i$ 
(note that $|j_{i'}-j_i|\geq 2$ since $F_{j-1}$ and $F_{j}$ never occur at the same time).
As each outcome occurs with probability $1/2^{r}$, the probability that some node rejects in Steps~4 or~5 is at least
\[
\frac{1}{2^{r}}\cdot \left[ (1/4)\cdot 2^{r}\right] \cdot \frac{1}{\lfloor r/2\rfloor} 
\sum_{j=1}^r \Pr[E_j|F_j] 
\; \geq \; \frac{1}{2r}\sum_{j=1}^{r}\Pr[E_j|F_j]
\; \geq \; \frac{1}{2r}\cdot \frac{1}{21r}= \frac{1}{42r^2},
\]
where the second last inequality comes from Lemma~\ref{lem:soundness}.

\subsection{Proof of Theorem~\ref{th:general-one-two-sided}}\label{sec:proof}

So far, we have shown that the protocol ${\cal P}_\pi$ has a completeness parameter 
very close to $1$, but high soundness error.
To establish Theorem~\ref{th:general-one-two-sided}, we need to reduce the soundness error without degrading the completeness too much. This  is achieved 
via a form of parallel repetition of ${\cal P}_\pi$, by taking the logical conjunction of the outputs 
obtained by repetitions. The  protocol resulting from $k$ repetitions of ${\cal P}_\pi$ is denoted by ${\cal P}_\pi[k]$, and works as follows.

\begin{center}
	\fbox{
		\begin{minipage}{13cm}
			{\bf Protocol ${\cal P}_\pi[k]$: Soundness reduction of Protocol ${\cal P}_\pi$}
			\begin{enumerate}
				\item If $f(x,y)=1$ then the prover sends the $k$ quantum registers $R_{j,i}$ ($i=1,\ldots,k$), each of which has a state $|h_x\rangle$ as certificate, to each of the intermediate nodes $v_j$, $j\in\{1,\ldots,r-1\}$. 
				\item The left-end node $v_0$ prepares the $k$ quantum registers $R_{0,i}$, each of which has $|h_x\rangle$.  
				\item For every $j=0,\dots,r-1$, the node $v_j$ chooses a $k$-bit string $b_{j,1}\cdots b_{j,k}$ uniformly at random, 
				and sends $R_{j,i}$, together with the index~$i$, to its right neighbor $v_{j+1}$ whenever $b_{j,i}=0$. 
				\item For every $j=1,\ldots,r-1$ and for every $i=1,\dots,k$, if the node $v_j$ receives an index~$i$, 
				and if $b_{j,i}=1$, then $v_j$ performs the SWAP test on $(R_{j-1,i},R_{j,i})$. 
				Node $v_j$ rejects whenever at least one of the performed SWAP tests rejects, 
				and it accepts otherwise.  
				\item 
				If the right-end node $v_r$ receives an index~$i\in\{1,\dots,k\}$, 
				then it performs the POVM measurement $\{M_{y,1},M_{y,0}\}$ corresponding to $\pi$ applied to the state in $R_{r-1,i}$. 
				Node~$v_r$ rejects if at least one of the performed POVM measurements rejects, and it accepts otherwise.
			\end{enumerate}
		\end{minipage}
	}
\end{center}

Protocol ${\cal P}_\pi[k]$ has completeness $(1-\varepsilon/2)^k$, that is, the completeness has not reduced much whenever $\varepsilon$ is small. 
By a similar analysis of standard error reduction techniques for quantum Merlin-Arthur games as the analysis in~\cite{AN02arXiv,KSV02book}, one can show that ${\cal P}_\pi[k]$ 
has soundness $(1-1/42r^2)^k$. By choosing $k=84 r^2$, 
the resulting protocol ${\cal P}_\pi[k]$
has completeness $1-42 r^2\varepsilon$ and soundness error $(1/e)^2<1/3$, 
while the size of the certificates is $O(r^2 q)$ qubits, 
and the length of the message exchanged between neighbors is $O(r^2 (q+\log r))$ 
(where the additional term $\log r$ comes from the index to the right neighbor in Step~3 
of ${\cal P}_\pi[k]$).

Theorem~\ref{th:general-one-two-sided} can now be easily derived from the above analysis.
For the first part of the theorem, where $f$ is having a one-sided-error one-way protocol $\pi$, simply use the protocol ${\cal P}_\pi$ from Section~\ref{sec:protocol} with $\varepsilon=0$. The result then follows from the analysis of Section~\ref{sec:analysis} and from the above discussion about soundness reduction.

For the case of second part of the theorem, where $f$ is having a two-sided-error one-way protocol, 
we repeat the protocol $\pi$
for $O(\log (n+r) )$ times 
and using \emph{majority voting} 
to get a protocol that correctly computes the value of the function with probability at least $1-1/42 n^{c}r^2$.
The protocol $\pi$ of Section \ref{sec:protocol} can thus be chosen with 
$\varepsilon=1/42 n^{c}r^2$, with message size $O(q\log (n+r))$. 
The result then follows similarly. 
\qed

\section{Certifying Equality in General Graphs}\label{sec:general-graphs}

We now extend our protocol for checking equality between $n$-bit strings $x_1,\dots,x_t$ stored at $t\geq 2$~distinct nodes $u_1,\dots,u_t$ of a connected simple graph~$G$. We first show how to reduce the problem to trees of a specific structure, and then present a protocol for trees. 

\subsection{Reduction to Trees}
Let $G=(V,E)$ be a connected simple graph, and let $u_1,\dots,u_t$ be $t\geq 2$~distinct nodes of~$G$. Assume, without loss of generality, that $u_1$ is the most central node among them, i.e., it satisfies 
$
\max_{i=1,\dots,t}\mathsf{dist}_G(u_1,u_i) = \min_{j=1,\dots,t}\max_{i=1,\dots,t}\mathsf{dist}_G(u_j,u_i). 
$
Let $r=\max_{i=1,\dots,t}\mathsf{dist}_G(u_1,u_i)$ be the \emph{radius} of the $t$ terminals $u_1,\dots,u_t$. 
We construct a tree $T$ rooted at $u_1$, that has all terminals as leaves, maximum degree $t$ and depth at most $r+1$ (see Figure~\ref{fig: tree ronstruction}).
To this end, start with a BFS tree $T'$ in $G$, rooted at $u_1$. Truncate the tree at each terminal $u_i$ that does not have any terminal as successors, thus limiting the depth to $r$ and the degree to~$t$. 
For every terminal $u_i$ that is not a leaf, including $u_1$, replace $u_i$ with a node $u'_i$, and connect $u_i$ to $u'_i$ as a leaf, where the input $x_i$ stays at $u_i$ --- this guarantees that all inputs are now on leaves, the same degree bound holds, and the depth is increased by at most~1.

\begin{figure}[]
	\begin{center}
		\includegraphics[scale=.75]{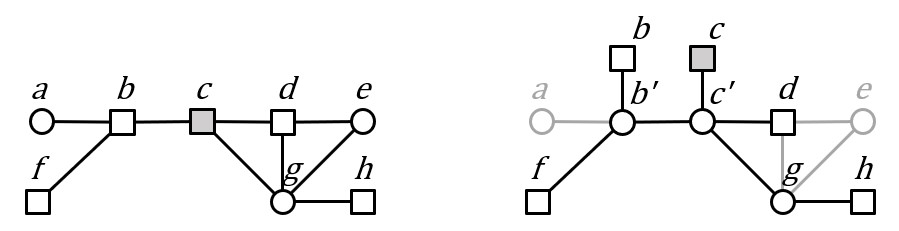}
		\caption{
			The construction of a spanning tree: a graph $G$ (left) and its corresponding spanning tree $T$ (right). 
			Terminals are marked by squares, and $c$ is the root.
			Node $b$ (resp.~$c$), which was a terminal, is replaced by a 
			non-terminal node $b'$ (resp.~$c'$), while the other terminals are leaves in the tree so they remain unmodified.}
		\label{fig: tree ronstruction}
	\end{center}
\end{figure}

While $T$ is not a sub-tree of $G$, we can easily emulate an algorithm or a labeling scheme designed for $T$, in $G$ (specifically, in $T'$). To this end, every internal terminal $u_i$ in $T'$ simulates the behavior of $u_i$ itself, and also of $u'_i$. The following lemma is using classical assumptions of network computing (see, e.g., \cite{Peleg00}) and can be proved using standard techniques (see, e.g., \cite{KormanKP10}). We refer to the tree~$T$ in the construction described above.

\begin{lemma}\label{lem:treeT2}
	For any graph $G=(V,E)$ with nodes IDs taken in a range polynomial in~$|V|$, there is a deterministic distributed Merlin-Arthur protocol for the tree $T$ using certificates on $O(\log |V|)$ bits. 
\end{lemma}

The term \emph{deterministic} in the above lemma means that the verification process is deterministic, which implies perfect completeness and perfect soundness (i.e., soundness error~$0$).
Roughly speaking, in this protocol each non-tree node will have a (non-quantum) label indicating its distance from the tree, and each tree node will have as label its depth in the tree, the ID of its parent, and the ID of the root.

\subsection{Certifying Equality in Trees}

Based on 
our tree construction from a graph and Lemma~\ref{lem:treeT2}, we can restrict our attention to the case in which the $t$ terminals $u_1,\dots,u_t$, who hold the $n$-bit strings $x_1,\dots,x_t$, belong to a tree~$T$ rooted at~$u_1$, of depth equal to $r+1$, where $r$ is the radius of the terminals, with maximum degree~$t$, and with leaves $u_2,\dots,u_t$. Moreover, we assume that 
the root~$u_1$ itself is of degree~1 due to our tree construction.   
We present a distributed quantum Merlin-Arthur protocol for the equality function $\EQ^{t}_{n}$ in this setting, and hence prove our main result.

\begin{theorem}\label{thm:tree}
	There is a distributed quantum Merlin-Arthur protocol on $T$ for 
	$\EQ^{t}_{n}$ between $t$ terminals of radius~$r$, 
	with perfect completeness,  soundness~$1/3$, certificate size $O(t\, r^2\, \log n )$ qubits, and message length $O(t\, r^2\, \log (n+r))$ qubits. 
\end{theorem}

\begin{proof}
	Let $\pi$ be a one-way communication protocol for $\EQ_n$ transmitting $m=O(\log n)$~qubits such that, for every input pair $(x,y)$, if $x=y$ then $\pi$ outputs~$1$ with probability $1$ and if $x\neq y$ then $\pi$ outputs~$0$ with probability at least~$2/3$ (such a protocol exists, as mentioned in Section \ref{sec:prelim}).
	For input $(x,y)$, let $|h_x\rangle$ be the quantum message 
	from Alice to Bob in $\pi$, 
	and let $\{M_{y,1},M_{y,0}\}$ be the POVM measurement 
	done by Bob on $|h_x\rangle$ in $\pi$, 
	where $M_{y,1}$ corresponds to the measurement result $1$ (accept),  
	and $M_{y,0}$ corresponds to the measurement result $0$ (reject), 
	respectively. Our quantum Merlin-Arthur protocol is as follows.
	
	\begin{center}
		\fbox{
			\begin{minipage}{13cm}
				{\bf Protocol }${\cal P}(\EQ^t_n)$ {\bf for equality in trees} 
				\begin{enumerate}
					\item If $\EQ^t_n(x_1,\dots,x_t)=1$ then the prover sends 
					an $m$-qubit state equal to $|h_{x_1}\rangle$ 
					to each of the nodes $v$ that have no input. 
					\item For every $i\in\{1,\ldots,t\}$, 
					node $u_i$ prepares the $m$-qubit state $|h_{x_i}\rangle$.
					\item Every non-root node $v$ of the tree 
					chooses a bit $b_v$ uniformly at random. 
					If $b_v=0$, then $v$ sends its $m$-qubit state to its parent in~$T$. 
					\item For every non-terminal node $v$, 
					if $v$ receives a state from the children, 
					and if $b_v=1$, then $v$ performs the SWAP test 
					on the $2m$-qubit state that consists of the $m$-qubit state 
					received from the prover and an $m$-qubit state 
					received from the children, which is chosen uniformly 
					at random if he/she receives multiple $m$-qubit states from the children, 
					and accepts or rejects accordingly. Otherwise, $v$~accepts. 
					\item If the root node $u_1$ receives a state from its children, 
					then $u_1$ performs the POVM measurement $\{M_{x_1,1},M_{x_1,0}\}$ 
					on an $m$-qubit state received from the children, 
					which is chosen uniformly at random 
					if he/she receives multiple $m$-qubit states from the children,  
					and accepts or rejects accordingly. Otherwise, $u_1$~accepts. 
				\end{enumerate}
			\end{minipage}
		}
	\end{center}
	
	The perfect completeness trivially holds 
	since every local test yields acceptance with certainty.
	For the soundness, if $\EQ^t_n(x_1,\dots,x_t)=0$ then 
	there is a leaf $u_i$, $i>1$, with $x_i\neq x_1$. 
	Then, we can perform almost the same analysis as in Section~\ref{sec:QDPP}, 
	but for the path connecting $u_1$ and $u_i$ in~$T$. 
	The only difference is the probability that each local test occurs; 
	it is at least $1/4$ in the analysis of ${\cal P}_\pi$ done in Section~\ref{sec:QDPP}, 
	while it is at least $(1/4)\cdot(1/t)$ in the protocol ${\cal P}(\EQ^t_n)$ we are now considering, 
	as every non-terminal node $v_j$ or $u_1$ on the path chooses 
	the $m$-qubit state from the child on the path uniformly at random 
	from the multiple $m$-qubit states (possibly) sent from all the children.  
	Hence, ${\cal P}(\EQ^t_n)$ has soundness error $1-O(1/tr^2)$. 
	The proof is completed by performing $O(tr^2)$ parallel repetitions of ${\cal P}(\EQ^t_n)$ 
	for error reduction, similarly to the $O(r^2)$ parallel repetitions of ${\cal P}_\pi$ in Section~\ref{sec:QDPP}.
\end{proof}

\subparagraph{Remark.} 

Using Lemma~\ref{lem:treeT2}, we get that, up to adding $O(\log|V|)$ classical bits in the certificates of the nodes,  Theorem~\ref{thm:tree} can be extended to the case where the terminals are in a connected graph $G=(V,E)$. 


\section{Classical Lower Bounds}
\label{sec:classicLWB}

In this section, we show that non-quantum distributed Merlin-Arthur (\dMA) protocols for distant $\EQ$ require certificates of linear size. In fact, we establish a more general lower bound which applies to all functions $f$ with large fooling set, even using shared randomness. 
In addition, the bound holds for settings which allow the graph nodes to have multiple communication rounds among them, after receiving the certificates and before deciding if they finally accept (see, e.g.,~\cite{FeuilloleyFHPP18,OstrovskyPR17}).

For the lower bound, it is sufficient to consider the path $v_0,\dots,v_r$ in which $v_0$ and $v_r$ are provided with inputs~$x$ and~$y$, respectively. 

\begin{theorem}\label{thm:clb}
	Let $r\geq 2\mu+1$, 
	and let $f(x,y)$ be any Boolean function 
	with a $1$-fooling set of size at least~$k$. 
	Let ${\cal P}$ be a classical Merlin-Arthur protocol for $f$ 
	in a path of $r$~edges,
	with $\mu$ rounds of communication among the nodes,
	shared randomness, 
	certificates of size $\lfloor \frac{1}{2\mu}\log (k-1)\rfloor$ bits, 
	and completeness~$1-p$. 
	Then ${\cal P}$ has soundness error at least~$1-2p$.
\end{theorem}

\begin{proof}
	Consider the path $v_0,\ldots,v_r$ with a fixed identifier assignment, and inputs $x$ and $y$ given to $v_0$ and $v_r$, respectively.
	Since $f$ has large $1$-fooling set but small certificates, 
	there exist two distinct pairs of ``fooling'' inputs 
	that have the same certificate assignments on the~$2\mu$ node $v_1,\ldots,v_{2\mu}$.
	That is, we can fix two input pairs $(x,y)$ and $(x',y')$,
	with $f(x,y)=f(x',y')=1$, and, w.l.o.g., $f(x,y')=0$, 
	with corresponding certificate assignments~$w$ and~$w'$, such that 
	\begin{equation*}\label{eq:clb}
		w(v_i)=w'(v_i) \; \mbox{for every} \; i\in\{1,\ldots,2\mu\},
	\end{equation*}
	where $w(v_i)$ and $w'(v_i)$ are the certificate assigned to the node $v_i$ in the assignments~$w$ and~$w'$, respectively.
	
	We interpret the outputs as Boolean values, where $\mbox{accept}=1$ and $\mbox{reject}=0$, and denote by $\mathsf{out}_i(x,y,w)$ the output of $v_i$ when the inputs are $x$ and $y$ and the certificate assignment is $w$.
	Since ${\cal P}$ has completeness $1-p$, we have
	\[
	\Pr_s \big [\bigwedge_{i\leq \mu}\mathsf{out}_i(x,y,w)=1 \wedge \bigwedge_{i\geq \mu+1} \mathsf{out}_i(x,y,w)=1 \big] \geq 1-p,
	\]
	and the same holds for $(x',y',w')$. Hence,
	\[
	\Pr_s \Big[\bigwedge_{i\leq \mu}   \mathsf{out}_i(x,y,w)=1  \Big] \geq 1-p,
	\mbox{ and } 
	\Pr_s \Big[\bigwedge_{i\geq \mu+1} \mathsf{out}_i(x',y',w')=1 \Big] \geq 1-p.
	\]
	
	The output $\mathsf{out}_i$ of every node $v_i$ is a function of its own identifier and certificate, the certificates of the nodes in its distance-$\mu$ neighborhood, and the public random string~$s$.
	In addition, 
	the outputs of $v_0,v_1,\ldots,v_\mu$ may also depend on the input~$x$ to $v_0$, and the outputs of $v_{r-\mu},\ldots,v_{r}$ may also depend on the input~$y$ to~$v_r$.
	Formally speaking, the outputs can also depend on the identifiers of the neighbors, but these are fixed given the node's identifier, so we ignore them henceforth.
	
	Let $w''$ be the certificate assignment defined by
	\begin{equation*}\label{eq: combined crtificates}
		\begin{array}{ll}
			w''(v_0)=w(v_0);\\
			w''(v_i)=w(v_i)=w'(v_i) &\mbox{ for } i\in\{1,\ldots, 2\mu\};\\
			w''(v_i)=w'(v_i) &\mbox{ for } i\in\{2\mu+1,\ldots,r\}.
		\end{array}
	\end{equation*}
	Consider the input assignment $(x,y')$ combined with the certificate assignment~$w''$.
	The definition of $w''$ implies that nodes $v_0,\ldots,v_{2\mu}$ receive the same certificates as in $w$, and thus nodes $v_0,\ldots,v_{\mu}$ cannot distinguish this form the input assignment $(x,y)$ with certificates assignment $w$.
	On the other hand, nodes $v_1,\ldots,v_r$ receive the same certificates as in $w'$, so the nodes $v_{\mu+1},\ldots,v_r$ cannot distinguish this form the input assignment $(x',y')$ with certificates assignment $w'$.
	
	A union bound finishes the proof:
	\begin{align*}
		&\Pr_s \Big [  	\bigwedge_{i\leq \mu}   \mathsf{out}_i(x,y',w'')=1 
		\wedge  		\bigwedge_{i\geq \mu+1} \mathsf{out}_i(x,y',w'')=1 
		\Big] \\
		=  	&\Pr_s \Big [  	\bigwedge_{i\leq \mu}   \mathsf{out}_i(x,y,w)=1 
		\wedge  		\bigwedge_{i\geq \mu+1} \mathsf{out}_i(x',y',w')=1 
		\Big] \\
		\geq   &1- \Pr_s \Big [ \neg \bigwedge_{i\leq \mu}   \mathsf{out}_i(x,y,w)=1 \Big] 
		- \Pr_s \Big [ \neg \bigwedge_{i\geq \mu+1} \mathsf{out}_i(x',y',w')=1  \Big] 	 \\
		\geq & \; 1- 2p. 
	\end{align*}
	That is, we found a certificate assignment $w''$ for the input $(x,y')$ which makes all node accept with probability at least~$1-2p$, even though $f(x,y')=0$. 
	Hence, the soundness error is at least~$1-2p$, as claimed.
\end{proof}

Since $\EQ_n$ has a $1$-fooling set of size $2^n$, the corollary below follows directly from Theorem~\ref{thm:clb}. The case $r=3$ and $\mu=1$ gives a linear lower bound for $\dMA$ protocols on a 4-node path.

\begin{corollary}\label{cor:LB}
	For every $r\geq 2\mu+1$, every distributed (classical) Merlin-Arthur protocol for $\EQ^2_n$ with $\mu$ rounds of communication among the nodes
	in the path of $r$ edges with certificates of size at most $\lfloor (n-1)/2\mu\rfloor$ bits, and completeness $1-p$ has soundness error at least~$1-2p$.
\end{corollary}

The requirements for a good protocol is to have a high completeness (i.e., small value for~$p$,  ideally $p=0$) and a reasonably small soundness error. Corollary \ref{cor:LB} precisely shows that for the equality function such protocols cannot exist in the classical setting unless the certificate size is linear in $n$.

\subparagraph{Remark.} 

The completeness-soundness gap of Theorem~\ref{thm:clb} is optimal in general, in the sense that it 
cannot be improved for $\EQ^2_1$, i.e., distant equality between two input bits. 
Consider the following protocol ${\cal P}$ for $\EQ^2_1$, on input $(x,y)\in\{0,1\}\times\{0,1\}$. It uses a shared random variable $X\in \{-1,0,1\}$
with $\Pr[X=0]=\Pr[X=1]=p$, and $\Pr[X=-1]=1-2p$. 
In the case $X=-1$, all the nodes accept. In the case $X\in\{0,1\}$, $v_0$~accepts whenever $X=x$,
$v_r$~accepts whenever $X=y$, and all the other nodes accept.  
If $x=y$, the probability that all the nodes accept is $1-2p+(1/2)\cdot (2p)=1-p$.
If $x\neq y$, then either $v_0$ or $v_r$ systematically rejects for $X\neq -1$, and thus  
the probability that all the nodes accept is $1-2p$.


\section{Conclusion}

In this paper, we extended the notion of randomized proof-labeling scheme to the quantum setting. We showed the efficiency of distributed quantum certification mechanisms by designing a distributed quantum Merlin-Arthur protocol for $\EQ_n^t$ between $t$ parties spread out in a graph, using certificates and messages whose size depend logarithmically on~$n$, the size of the data. This is in contrast to classical distributed Merlin-Arthur protocols, which require certificates of size linear in~$n$, even when messages of unbounded size can be used. 
Our result was obtained by using an interesting property of the SWAP test: it can be applied for checking proximity properties between reduced states.

Which other Boolean predicates on labeled graphs, beyond equality, could take benefit from quantum resources for the design of compact distributed certification schemes is an intriguing question. Theorem \ref{th:general-one-two-sided} gives a partial answer on the path. A complete answer to this question would significantly help improving our comprehension of  the power of quantum computing in the distributed setting.

\bibliography{dQMA-bib}


\appendix

\

\centerline{\Large\bf A P P E N D I X}

\section{Quantum Fundamentals}
\label{sec:QuantFund}

Here we summarize some notation and properties that are used in this paper.
As a terminology, we sometime identify physical concepts (such as pure states) 
and their mathematical representations (such as vectors). 

A {\em mixed state} on a complex Euclidian space ${\cal H}$ 
is considered as a representation of a probability distribution 
of {\em pure} quantum states (represented like $|\psi\rangle$ as vectors on ${\cal H}$).
If the quantum state in ${\cal H}$ is a pure state $|\psi_j\rangle$ with probability $p_j$, 
its mixed state is represented as the positive semi-definite matrix $\sigma=\sum_j p_j |\psi_j\rangle\langle\psi_j|$
(the symbol $\rho$ or $\sigma$ is often used for representing a mixed state). 
In particular, a (pure) quantum state $|\psi\rangle$ is written as the projector $|\psi\rangle\langle\psi|$ of rank $1$. 
If this state evolves by a unitary operation $U$,
the state changes into $U\sigma U^\dagger=\sum_j U|\psi_j\rangle\langle \psi_j|U^\dagger$. 
We call a mixed state simply a (quantum) state (as far as we do not care about the difference between pure and mixed states).

For any complex Euclidian spaces ${\cal H}_A$ and ${\cal H}_B$ and any matrix $M$ on ${\cal H}_A\otimes{\cal H}_B$, 
the reduced matrix on ${\cal H}_A$ obtained by tracing out on ${\cal H}_B$, denoted as $\tr_{B}(M)$ is   
\[
\tr_{B}(M)=\sum_{j} (I \otimes \langle j|)M(I\otimes |j\rangle),
\]
where $\{|j\rangle\}$ is the standard basis in ${\cal H}_B$ (in fact, it may be any orthonormal basis).
If $M$ represents a mixed state $\sigma$, $\tr_B(\sigma)$ is called 
the {\rm reduced state} on ${\cal H}_A$, which represents the locally visible state 
on ${\cal H}_A$ of $\sigma$ obtained by disregarding the part on ${\cal H}_B$ of $\sigma$.

A {\em POVM (positive operator valued measure)} on a complex Euclidian space ${\cal H}$ represents a general measurement 
on a quantum state on ${\cal H}$. 
In particular, a binary-valued POVM (which we use in this paper) on ${\cal H}$ is a set 
$\{M_0,M_1\}$ that consists of two positive semi-definite matrices $M_0$ and $M_1$ on ${\cal H}$ 
such that $M_0+M_1=I$. If a mixed state $\rho$ is measured by $\{M_0,M_1\}$,
the probability that the outcome with $M_j$ ($j=0,1$) is obtained is $\tr(M_j\rho)$. 

For any matrix $M$ on complex Euclidian space ${\cal H}$, 
the trace norm of $M$ is defined as
\[
\|A\|_{\tr}=\tr(\sqrt{M^\dagger M}).
\]
For any two mixed states $\rho$ and $\sigma$ in ${\cal H}$, 
the trace distance between $\rho$ and $\sigma$ is defined as
\[
\mathsf{dist}(\rho,\sigma)=\frac{1}{2}\|\rho-\sigma\|_{\tr},
\]
which satisfies that $0\leq \mathsf{dist}(\rho,\sigma)\leq 1$ (as $\mathsf{dist}(\rho,\sigma)$ becomes close to $0$, $\rho$ and $\sigma$ become close).  
In particular, for any pure states $|\psi\rangle,|\phi\rangle$, it holds that 
\[
\mathsf{dist}(|\psi\rangle\langle\psi|,|\phi\rangle\langle\phi|)=\sqrt{1-|\langle\psi|\phi\rangle|^2}.
\]
An important characterization of the trace distance is
\[
\mathsf{dist}(\rho,\sigma)=\max_{M} \tr(M(\rho-\sigma)),
\]
where the maximization is taken over all positive semi-definite matrices $M$ 
such that $M\leq I$. This characterization means 
that $\mathsf{dist}(\rho,\sigma)$ is equal to the difference between 
the probability that the ``best'' POVM measurement $\{M,I-M\}$, 
for distinguishing $\rho$ and $\sigma$, on the state $\rho$ gives the outcome with the POVM element $M$, 
$\tr(M\rho)$, and the probability that the same POVM measurement on $\sigma$ 
gives the outcome with $M$, $\tr(M\sigma)$. 

For any two mixed states $\rho$ and $\sigma$ in ${\cal H}$, 
the fidelity between $\rho$ and $\sigma$, denoted as $F(\rho,\sigma)$, 
is another measure of their closeness, and it holds that $0\leq F(\rho,\sigma)\leq 1$ 
(as $F(\rho,\sigma)$ becomes close to $1$, $\rho$ and $\sigma$ become close).   
For this measure, we need only the following equality, which is found in Ref.~\cite{Wat18book} for instance. 

\begin{lemma}[Corollary 3.23 in Ref.~\cite{Wat18book}]\label{lem:tracenorm}
	Let $|\psi\rangle$ and $|\phi\rangle$ be two pure states on ${\cal H}\otimes {\cal H}'$ 
	for complex Euclidian spaces ${\cal H}$ and ${\cal H}'$. It holds that 
	\[
	F( \tr_{\cal H'}(|\psi\rangle\langle\psi|), \tr_{\cal H'}(|\phi\rangle\langle\phi|) ) 
	= \| \tr_{\cal H}( |\phi\rangle\langle\psi| ) \|_{\tr}. 
	\]
\end{lemma}

Let ${\cal H}_1$ and ${\cal H}_2$ be two complex Euclidean spaces consisting of $m$ qubits for each 
(namely, each of the space has the standard basis states $\{|x\rangle\mid x\in\{0,1\}^m\}$).
Then, ${\cal H}={\cal H}_1\otimes {\cal H}_2$ can be written as the direct sum of 
the symmetric space ${\cal H}_S$ and the antisymmetric subspace ${\cal H}_A$,
Here, the symmetric space ${\cal H}_S$ is the subspace of ${\cal H}$ 
spanned by the states $|\psi\rangle$ in ${\cal H}$ such that $S|\psi\rangle=|\psi\rangle$ where $S$ is the swap operator 
defined by $S(|j_1\rangle|j_2\rangle)=|j_2\rangle|j_1\rangle$. 
The antisymmetric subspace is the subspace $|\psi\rangle$ in ${\cal H}$ such that 
$S|\psi\rangle=-|\psi\rangle$. Note that ${\cal H}_A$ is the orthogonal complement of ${\cal H}_S$.
The dimensions of ${\cal H}_S$ and ${\cal H}_A$ are $M(M+1)/2$ and $M(M-1)/2$, respectively, where $M=2^m$. 
See Ref.~\cite{Wat18book}, for instance, for more information on (bipartite as well as multipartite) symmetric states and anti-symmetric states.

\section{Elementary Lemma on Probability}\label{sec:ProbFund}

For any events $A,B$, we denote the complementary event of $A$ by $\overline{A}$ or $\neg A$, 
the sum event of $A$ and $B$ by $A\vee B$, and the product event of $A$ and $B$ by $A\wedge B$. 
The following lemma is basic on probability, while we give it with the proof for the self-containment.

\begin{lemma}\label{lem:probability}
	Let $A_j$ ($j=1,2,\ldots,n$) be an event. Then, following holds.
	\begin{itemize}
		\item $\Pr[A_1\wedge A_2\wedge \cdots \wedge A_n]\leq \frac{1}{n}\sum_{j=1}^n \Pr[A_j]$
		\item $\Pr[A_1\vee A_2\vee \cdots \vee A_n]\geq \frac{1}{n}\sum_{j=1}^n \Pr[A_j]$
	\end{itemize}
\end{lemma}

\begin{proof}
	We show the first item by induction. The base case trivially holds. Assume that the case $n-1$ holds.
	Then,
	\begin{align*}
		&\!\!\!\! \Pr[A_1\wedge \cdots \wedge A_n]\\
		&=\frac{n-1}{n} \Pr[A_1\wedge\cdots \wedge A_{n-1}]\Pr[A_{n}|A_1\wedge\cdots \wedge A_{n-1}]+\frac{1}{n} \Pr[A_n]\Pr[A_1\wedge\cdots \wedge A_{n-1}|A_n]\\
		&\leq \frac{n-1}{n} \Pr[A_1\wedge\cdots \wedge A_{n-1}]+ \frac{1}{n} \Pr[A_n]\\
		&\leq  \frac{n-1}{n} \cdot \frac{1}{n-1} \sum_{j=1}^{n-1} \Pr[A_j] + \frac{1}{n} \Pr[A_n]\\
		&=\frac{1}{n}\sum_{j=1}^{n} \Pr[A_j],
	\end{align*}
	where the last inequality comes from the assumption. Thus, the case $n$ holds, and the induction is completed.
	
	The second item is proved by
	\begin{align*}
		\Pr[A_1\vee \cdots \vee A_n] 
		&= 1-\Pr[\overline{A_1}\wedge\cdots\wedge \overline{A_n}]\\
		&\geq 1-\frac{1}{n}\sum_{j=1}^n \Pr[\overline{A_j}]\\
		&= \frac{1}{n}\sum_{j=1}^n \Pr[A_j],
	\end{align*}
	where the inequality comes from the first item.
\end{proof}

\end{document}